\documentclass[12pt,letterpaper]{article} 

\usepackage{setspace}
\usepackage{amsmath,amssymb,amsthm,amsfonts,fullpage,mathrsfs} 
\usepackage{graphics} 
\usepackage[page,header]{appendix}	
\usepackage{graphicx} 
\usepackage{natbib}
\usepackage{tikz,subfigure,multirow}

\usepackage{xy}   
\usepackage{stmaryrd} 
\usepackage{xcolor}  

\newtheorem{thm}{Theorem}

\usepackage{url}

\newtheorem{lem}{Lemma} 
\newtheorem{prop}{Proposition} 
 
\newtheorem{claim}{Claim}

\newcommand{\Gp}[2]{G^+(#1,#2)}
\newcommand{\Gq}[2]{G^-(#1,#2)}

\newcommand{\SO}{f(S)}

\newcommand\extR{\overline{\mathbb{R}}}
\newcommand\extlt{\mathrel{\overline{<}}}
\newcommand\extSigma{\overline{\Sigma}}
\newcommand\med{\textsc{med}}

\newcommand\tauprod{\tau^N}

\makeatletter
\def\@xfootnote[#1]{%
  \protected@xdef\@thefnmark{#1}%
  \@footnotemark\@footnotetext}
\makeatother

\usetikzlibrary{patterns}

\begin{document}

\title{The Limits of Tolerance} 
\author{Alan D. Miller\thanks{Faculty of Law, Western University, 1151 Richmond Street, London, Ontario N6A 3K7, Canada. Email: alan.miller@uwo.ca. Web: http://alandmiller.com. This research was undertaken, in part, thanks to funding from the Canada Research Chairs program. I would
like to thank Herv\'{e} Moulin and two anonymous referees, as well as Itai Arieli, Veronica Block, Chris Chambers, H\"{u}lya Eraslan, Nicholas Gravel, Matthew Kovach, C\'{e}sar Martinelli, Elisa Miller, Marcus Pivato, Ben Polak, Clemens Puppe, Jacob Schwartz, Matthew Spitzer, Bruno Strulovici, William Thomson, Jianrong Tian, John Weymark, Bill Zame, participants at the 32nd Annual Conference of the European Association of Law and Economics, the 29th Annual Meeting of the American Law and Economics Association, the 7th Israeli Game Theory Conference, the 2022 Meeting of the Society for Social Choice and Welfare, the GRASS XVII meeting, the 2024 SAET conference, the Theory and Experiment at McMaster Workshop, the 2024 annual conference of the Israeli Law and Economics Association, the 2024 STILE workshop, the 2024 meeting of the Canadian Law and Economics Association, LET XIII, the 2025 AsLEA Annual Conference, the 2025 World Congress of the Econometric Society, and seminars at ETH Z\"{u}rich, the University of Southern Denmark, Seoul National University, the University of Rochester,  the University of Toronto, the University of Arizona, Ben Gurion University, Bar Ilan University, and Western University.}}

\maketitle

\begin{abstract}
I propose a model of aggregation of intervals relevant to the study of legal standards of tolerance. 
Using this model, I introduce the family of endpoint rules, which aggregate intervals according to the $p$-th lowest lower endpoint and the $q$-th highest upper endpoint.
This family includes extreme tolerance (allowing anything permitted by anyone) and a form of majoritarianism (the median rule).
I show that the family of endpoint rules is characterized by the combination of anonymity, responsiveness, continuity, and ordinal neutrality, as well as by the combination of anonymity, strategyproofness, and translation equivariance.
\end{abstract}

\section{Introduction}

Common law legal systems often rely on community standards, a legal concept according to which actions are judged by the standards of society.
Community standards are ubiquitous in the common law. 
In the law of accidents, negligence is determined according to the standard of the reasonable person.\footnote{``We apply the standards which guide the great mass of mankind in determining what is proper conduct of an individual under all the circumstances and say that he was or was not justified in doing the act in question.'' \textit{Osborne v. Montgomery}, 203 Wis. 223, 231 (1931). There is a considerable debate about the extent to which courts should and do follow this general rule; an alternative approach treats the reasonable person as a normative standard. For more, see \citet{mp:2012}.} 
In contract law, determinations of good faith are made by reference to ``community standards of decency, fairness or reasonableness.''\footnote{Restatement (Second) of Contracts, \S~205 cmt. a. For more, see \citet{mp:2012b}.} 
In the United States, speech that is obscene according to ``contemporary community standards'' may be criminalized notwithstanding constitutional protection of free speech.\footnote{\textit{Miller v. California}, 413 U.S. 15 (1973). See also \citet{am:2013,am:2024}.} 
These standards are used in cases where fixed standards of behavior are undesirable because the standard is hard to define or expected to change over time.\footnote{Community standards are also found outside of the common law. For example, the Internal Revenue Code exempts from the prepaid interest rule points paid on the mortgage of a principal residence, provided that ``such payment of points is an established business practice in the area in which such indebtedness is incurred.'' 26  U.S.C. \S \;461(g)(2). The Statute of the International Court of Justice requires the court to apply, among other sources, ``international custom, as evidence of a general practice accepted as law'' when resolving international disputes. \textit{Statute of the International Court of Justice}, art. 38,  $\P$   1(b).} 
They align the law with local needs and values.\footnote{For example, the reasonable person standard aligns negligence law with contemporary understandings of fault. The good faith standard aligns contract law with the language, practices, and expectations of specific merchant communities, making it easier for members of these communities to draft private agreements. The ``contemporary community standards'' of obscenity align constitutional and criminal law with moral views that may differ across time and place.}

I introduce a model of community standards in which standards are represented by intervals of the real line. 
A set of individual standards is aggregated to form a community standard. 
Normative axioms are introduced that represent principles important in legal decision-making. 
These axioms are used to characterize a new family of aggregation rules: the \textit{endpoint rules}.

To illustrate, consider the law of defamation, which provides a right to sue for damages when someone else publishes speech that harms one's reputation in the eyes of the community \citep[see][]{mp:2013}. 
The person may claim, for example, that the statement attributed to her extreme political views that she does not in fact hold, and that the association with those views has harmed her reputation.\footnote{There are, of course, many political dimensions; nonetheless, the language of political left and political right suggests that a model with one dimension may have some explanatory power. One may also restrict attention to beliefs along a single dimension, such as views relating to the use of military force.} 
Members of society each believe that respectable people are neither on the extreme left nor the extreme right, but they may have different beliefs about what is extreme in each direction. 
One may think of these as cutoffs defining an interval of reputable political beliefs. 
The court must determine whether the speech is defamatory, and this determination must be based on the beliefs of the members of the community.

Endpoint rules use these cutoffs---the points beyond which an action is deemed unreasonable, or a statement defamatory---to determine the limits of social tolerance. 
The formal structure of these rules will be described below, but the basic concept of these rules is that the bounds of tolerance are set by eliminating the most extreme cutoffs. 
In the context of defamation, this implies that a court will not merely look to see whether individuals consider the speech to be extreme, but rather, whether they consider it to be extreme for the same reasons. 
The rules eliminate the possibility that the statement harms the plaintiff's reputation because some members of the community consider it to be too far to the left, and others consider it to be too far to the right.

I now proceed to describe the model. 
The first element of the model is the set of alternatives, which is isomorphic to the real line. 
Next, there are agents; these represent the members of the community.\footnote{The composition of the community may vary according to the claim. For example, in the United States, one may claim that a statement is defamatory ``in the eyes of a
substantial and respectable minority,'' while in the United Kingdom, the statement must be defamatory ``in the eyes of the general community.'' For more, see \citet{mp:2013}. } 
Each agent has a non-empty set of allowable alternatives; this set is assumed to be bounded and convex.\footnote{To simplify results, I assume that judgments take the form of open sets.} 
In other words, judgments take the form of intervals of the real line. 
These individual judgments are then aggregated to form the community standard.

The focus of this paper is on the method through which these judgments are combined. 
A set of axioms is defined and used (in various combinations) to characterize the family of endpoint rules. 
These rules aggregate the lower and upper endpoints separately, in a way that guarantees that the aggregate will be an interval. 

Endpoint rules are parametrized by two positive integers, $p$ and $q$, such that the sum of the two parameters is not greater than $n+1$, where $n$ is the number of agents. 
As each agent has a judgment that takes the form of an interval, I define the $p,q$-th endpoint rule as the one that defines the aggregate set to be the interval defined by the $p$-th lowest lower endpoint and the $q$-th highest upper endpoint. 
The subclass of rules where $p=q$ is the family of ``symmetric'' endpoint rules. 
 
The  family of endpoint rules includes the ``maximal rule,'' in which $p=q=1$ \citep[see][]{am:2009,gtw:2017}, and the ``median rule,'' in which $p=q=\lfloor\frac{n+1}{2}\rfloor$ \citep[see][]{am:2009,vb:2010,fc:2011,gtw:2017}.
In practice, there may not be a well-defined median judgment.
Endpoint rules provide one answer to this problem: as in \citet{jl:2007}, the median rule is well-defined even though a median judgment may not exist.\footnote{For example, consider three judgments: (1,6), (2,3), and (4,5). None of these judgments is a median, but the median rule is defined and yields (2,5). Majority rule does not lead to a well-defined interval, as a majority supports the points in (2,3) and (4,5) but not those in [3,4].}

Endpoint rules aggregate judgments according to their endpoints, and not in a pointwise manner.
This is important because pointwise aggregation (for example, according to majority rule) will not necessarily result in a well-defined interval and, consequently, may lead to incoherent outcomes.
A further feature of these rules is that they aggregate the endpoints independently. 
That is, the aggregate left endpoint is determined without reference to the individual right endpoints, and vice versa.

To return to the example, a court may ask a jury to determine whether a statement harms the reputation of the plaintiff in the eyes of the community. 
The jury may unanimously agree that the statement does in fact harm the plaintiff's reputation. But they may do so for different reasons; some may consider it too liberal, and others, too conservative. 
As such, it is possible that all statements would be judged defamatory. 
But society cannot function if all statements are considered to be defamatory. 
Under majority aggregation, one can reach an equally perverse result: it may be possible that calling someone a moderate is defamatory, but calling them a liberal or a conservative is not.

In practice, a court will only consider the case presented to it, and not all possible cases. 
As a consequence, an observer may not receive enough information to tell whether similar cases would have been decided in a coherent manner.
But the problem of incoherence still remains, and the decision would still be fundamentally arbitrary.

Endpoint rules solve this problem by separating the decision of whether an action is permissible into two questions: 
First, is there a `lesser' permissible action? 
Second, is there a `greater' permissible action?
That is, one first checks to see whether enough people believe that there is a statement that could be made, referring to that individual as being further to the left, that would not be defamatory. 
These people need not agree on what that statement would be; each need only believe that such a statement exists.
One then checks to see whether enough people (though not necessarily the same number as before) believe that a statement placing that individual further to the right would be non-defamatory. 
If the answer to both of these questions is yes, the claim of defamation would fail.

I characterize the class of endpoint rules using several axioms. 
The \textit{responsiveness} axiom requires the aggregation rule to respond appropriately to changes in the individual judgments. 
If the individual intervals change, and each new interval includes the prior one, then responsiveness requires the new aggregate interval to include the prior aggregate interval. 
The \textit{anonymity} axiom requires the aggregate to be independent of the names of the agents, so that the aggregate choice would not change were two agents to trade their standards between themselves.
The \textit{continuity} axiom requires that small changes in individual standards lead to only small changes in the community standard.

To understand the fourth and fifth axioms,  recall that the model imposes a structure on the set of alternatives, in that the set of alternatives is isomorphic to the real line. 
These axioms are motivated by the idea that points on the political spectrum are not inherently special. 
The political center  is a function of individual beliefs and cannot be objectively defined; one may note that views on civil rights, same-sex marriage, capital punishment, gun control, and democracy have changed over time. 
These views also have no meaningful cardinal relationship. 

Nonetheless, the political spectrum has a natural structure. 
One important property of this structure is \textit{betweenness}; on the political spectrum, a centrist is objectively in between a liberal and a conservative \citep[see][for more on betweenness]{np:2002,np:2006}. 
A second property is \textit{direction}; in some contexts one may wish to treat the left differently from the right. 
Such a property may be relevant in other contexts, such as when trying to determine the reasonableness of highway driving speeds in negligence.\footnote{In the context of speeds, one may think of the endpoints as being drawn from $(0,\infty)$, to reflect the idea that one should not park in the middle of the highway. Note that, while speeds themselves have a cardinal relationship, the reasonableness of speeds may not.} 
Here, the directions of ``high'' and ``low'' are not interchangeable.

The next two axioms require the aggregation rule to treat the real numbers as mere labels.
That is, a relabeling of the individual standards, in a way that preserves the relevant structural properties of the real line, leads to the same relabeling of the community standard.

The \textit{ordinal neutrality} axiom applies this requirement to relabelings of the real line that preserve both betweenness and direction. 
This axiom implies that the cardinal properties of the real line should be disregarded. 
Ordinal neutrality is the analogue, for interval aggregation, of the ordinal covariance axiom of \citet{cc:2007}.
 
The \textit{direction neutrality} axiom applies this requirement to a specific relabeling of the real line: the reflection about the origin.  
It restricts the aggregation rule from placing any special meaning on the left or on the right.
When combined, ordinal and direction neutrality apply this requirement to relabelings that preserve betweenness, regardless of whether they preserve direction.

The family of endpoint rules is characterized by four axioms---responsiveness, anonymity, continuity, and ordinal neutrality.
The addition of the fifth axiom---direction neutrality---then yields a characterization of the symmetric endpoint rules.

\subsection{Strategyproofness}

The interval aggregation problem described above does not use the concept of preference. 
Individuals' judgments represent the individuals' beliefs about which actions are acceptable, and not their preferences over policy \citep[for more on this distinction, see][]{ks:1986}. 
The model takes the judgments as given, and does not ask where they come from. 

The problem of interval aggregation exists in contexts other than the determination of community standards in legal settings. 
A panel of experts may need to provide an aggregate opinion in the form of a range \citep[see][]{cc:1993,fc:2011}. 
A legislature may need to delegate authority---in the form of a range of allowable alternatives---to an administrative agency \citep[see][]{eo:1994,eo:1996,eo:1999,eo:2008,hs:2002,hs:2006,sg:2009,gp:2012,ck:2014}. 
A developer of a large language model may need to evaluate whether the output of the model falls within the range that people find acceptable---for example, whether an answer leans too far to the political left or right---by aggregating the judgments reported by individual evaluators \citep[see][]{wgh:2025,glp:2022,mpd:2025}.
In these other contexts, individuals may be able to manipulate the result by misrepresenting their beliefs, whether the result be the aggregate expert opinion, the grant of authority to the administrative agency, or the assessment of the large language model.
To examine misrepresentation, I consider the implications of a \textit{strategyproofness} assumption \citep[see][]{df:1961}.

To study strategyproofness it is necessary to restrict the class of allowable preferences \citep[see][]{ag:1973,ms:1975}.
It is known that, when choosing a single alternative from a single issue dimension, a voting rule can be strategyproof and non-dictatorial if preferences are single-peaked \citep[see][]{hm:1980}.\footnote{These preferences are distinct from the ``single-plateaued'' preferences studied in \citet{db:1998}, where the plateaus represent indifference between top-ranked objects.} 
\citet{vb:2010} and \citet{fc:2011} define a class of generalized single-peaked preferences according to which an interval is defined to be between two other intervals if its lower endpoint is between the lower endpoints of the other two, and if its upper endpoint is between the upper endpoints of the other two. 
Preferences are generalized single-peaked if there is (a) a unique interval that is preferred to all other intervals (called the `peak') and (b) any interval in between the peak and a third interval is weakly preferred to the third interval \citep[see][]{np:2006}. 
 An aggregation rule is strategyproof if each individual prefers to report her peak interval. 
 \citet{vb:2010} and \citet{fc:2011} show that the median rule is strategyproof.

I provide a full characterization of anonymous and strategyproof interval aggregation rules. 
To prove this characterization I first show that the strategyproofness axiom implies that the aggregation rule must aggregate the endpoints independently.
As a consequence, the aggregation of lower endpoints is essentially equivalent to the aggregation of single-peaked preferences on a single issue dimension as in \citet{hm:1980}. 
I then show that the results in that work can be used to characterize a family of rules analogous to Moulin's famous ``phantom voters'' characterization.

This family includes the endpoint rules as a special case.
I show that endpoint rules can be characterized by further adding a \textit{translation equivariance} axiom, which requires the aggregation rule to shift the aggregate interval by a constant when each individual interval shifts by that same constant. 
Translation equivariance is implied by ordinal neutrality; in this sense it may be thought of as a very weak form of neutrality.

An implication of this characterization is that ordinal neutrality is not necessary to support the use of endpoint rules.
Unlike ordinal neutrality, translation equivariance allows the rule to make use of the cardinal properties of the real line.
Averaging rules, which determine the endpoints through the use of averages  \citep[see][]{gtw:2017}, are not ordinally neutral, but are translation equivariant and, in some cases, anonymous.
However, averaging rules are not endpoint rules; when anonymous they are ruled out by the assumption of strategyproofness.

\subsection{Other literature}

There is a significant literature devoted to the study of opinion and judgment aggregation, starting with the pioneering works of \citet{ar:1951} and \citet{km:1952}.
The closest work on community standards is \citet{am:2013}, which differs in that standards in that work are arbitrary subsets of an unstructured set of alternatives, rather than intervals of the real line. 
The lack of an objective order leads to a near-impossibility result; \citet{am:2013} provides conditions under which an aggregate of individual standards will deem an action impermissible only when all individuals in the community consider it to be impermissible.\footnote{\citet{ac:2010} reaches a similar result in the context of  menu choice. Related results, using different sets of axioms, can also be derived from \citet{bm:1990} and \citet{np:2006}.}
As communities are generally understood to be large and diverse, nothing, in practice, would be forbidden. 
By contrast, the present work provides insight into how the structure of the real line can be exploited to construct more useful rules.

In some sense, endpoint rules engage in a form of issue-based (or premise-based) aggregation, where the issues are the questions of whether there is a `lesser' and `greater' permissible action, and the outcome is whether the action is permissible. The distinction between issue aggregation and outcome aggregation has long been studied in law \citep[see][]{ps:1947,ks:1986}. For a formal treatment, see \citet{lp:2002}.

There are a number of papers that study interval aggregation from an axiomatic perspective. 
\citet{am:2009} introduces a model of interval aggregation, along with the median rule and the maximal rule, and shows that these two rules are, respectively, the least and most permissive rules that satisfy responsiveness, anonymity, ``homogeneity'' (an axiom that requires the aggregation rule to preserve complete agreement), and strong neutrality, which is equivalent to the combination of the ordinal and direction neutrality axioms.\footnote{\citet{am:2009} is an earlier version of this paper.} 
As mentioned above, \citet{vb:2010} and \citet{fc:2011} introduce the class of generalized single-peaked preferences over intervals\footnote{For the concept of generalized single-peakedness, see \citet{np:2006}.} and show that the median rule is strategyproof when preferences are restricted to this class. More recently, \citet{ent:2022} considers different representations of intervals and shows that the only rules that consistently aggregate both endpoints and interval widths are those which use weighted averages.

Theorem~\ref{mainthm} is reminiscent of \citet{km:1952}, which characterizes majority rule using positive responsiveness,  anonymity, and neutrality. 
However, it should be noted that the responsiveness and neutrality axioms used in this paper are conceptually distinct from those used by May.
For generalizations of May's axioms in spatial environments, see \citet{bc:2015,bc:2017}.

There is a conceptual link between the results in this paper and those of \citet{cc:2007}, which characterizes quantile representations using ordinal covariance and monotonicity. 
Ordinal covariance is essentially ordinal neutrality, while monotonicity is closely related to responsiveness. 
Endpoint rules can be thought of as pairs of quantile rules, as each endpoint is chosen according to a quantile. 
A contribution of the present work is that endpoint rules allow for independent and consistent aggregation of the two endpoints.

\citet{ss:2003} introduces the family of ``consent rules'' in the context of group identification.
While the models and rules are quite distinct, one can draw an analogy between consent rules and endpoint rules in that both families exist within a spectrum that places liberalism and democracy at its ends.\footnote{It is important to note that the spectrum is two-dimensional in both cases; a rule may be liberal in one direction and less liberal in the other.}
The maximal rule may be thought of as the most liberal rule in that any action considered reasonable by at least one person is permitted by the rule. 
The median rule, on the other hand, may be thought of as the most democratic; here, the boundaries of acceptability are majoritarian. 

The results involving strategyproofness are related to those in \citet{bsz:1991}, which studies a model in which individuals must choose a subset of objects from some finite set. 
In particular, \citet{bsz:1991} uses anonymity, neutrality, strategyproofness, and voter sovereignty properties to characterize a mechanism known as ``voting by quota,'' in which objects are chosen if a quota is met; that is, if a large enough group of individuals wants those objects to be chosen. 
The main difference from the results in \citet{bsz:1991} is that, in this work, the objects are ordered, and individuals must choose convex subsets. 
As a consequence, there is no need for all objects, or even any object, to be supported by a quota, as an object will be chosen as long as enough people agree  that (a) some lesser object should be chosen and (b) some greater object should be chosen as well. 
They need not agree on the identity of the lesser or greater object, but only that some such object exists.
The axioms are different; the neutrality axioms in this paper do not apply to all permutations, but each instead preserves some of the features of the real line. 
In addition, the strategyproofness axiom relies on a different assumption regarding preferences.

The results involving strategyproofness are also related to those in \citet{bj:1983} and \citet{bms:1998}, which characterize strategyproof voting mechanisms when the set of alternatives is a subset of some Euclidean space. 
The relationship becomes apparent when one considers that intervals may be described as those elements $(x,y)\in\mathbb{R}^2$ for which $x<y$. 
In this context, the preferences defined in those works imply the generalized single-peaked preferences studied in this paper.

\section{Endpoint Rules}

Let $N\equiv\{1,\hdots{},n\}$, $n\geq 2$, be a finite set of agents, and let $\Sigma$ be the set of bounded open non-empty intervals of the real line. Elements of $\Sigma$ are referred to as \textit{standards}. 
I study aggregation rules $f : \Sigma^N \rightarrow \Sigma$, which map a profile of $n$ standards into a single standard. 

The first two axioms, responsiveness and anonymity, are straightforward.
An individual standard $T_i$ is weakly more tolerant than another standard $S_i$ if $S_i \subseteq T_i$; similarly, a profile of standards $T$ is weakly more tolerant than another profile $S$ if $S_i \subseteq T_i$ for all $i\in N$. Responsiveness requires that, when a profile of standards $T$ is weakly more tolerant than another profile $S$, the resulting community standard $f(T)$ must be weakly more tolerant than $\SO$.

\begin{description} \item [Responsiveness:] For all $S,T\in\Sigma^N$, if $S_i \subseteq T_i$ for all $i\in N$, then $\SO \subseteq f(T)$. \end{description}

For a profile $S\in\Sigma^N$ and a permutation $\pi$ of $N$, define $\pi S \equiv (S_{\pi(1)},\dots{},S_{\pi(n)})$.
Anonymity restricts the aggregate from being affected by a change of the agents' names.

\begin{description} \item [Anonymity:] For every permutation $\pi$ of $N$ and every $S\in\Sigma^N$, $\SO = f(\pi S)$. \end{description}

I introduce a basic continuity axiom.\footnote{I thank Jianrong Tian for pointing out a problem with the definition of continuity used in an earlier version of this paper.} For $x,x'',y,y''$ in $\mathbb{R}$, $x<x''<y<y''$, let $W(x,x'',y,y'')\subset\Sigma$ be the set of all intervals $(x',y')\in\Sigma$ such that $x < x' < x''$ and $y<y'<y''$.
Let $(\Sigma,\tau)$ be the topological space with the base $\{W(x,x'',y,y''):x<x''<y<y''\}$.
Let $(\Sigma^N,\tauprod )$ be the product topology.

\begin{description} \item[Continuity:] For all $Y\in\tau$, $f^{-1}(Y)\in\tauprod $.\end{description}

I provide two distinct neutrality axioms; each requires the aggregation rule to be equivariant with respect to relabelings of the real line.
Let $\Phi$ be the set of all strictly increasing bijections of the real line onto itself. For $X\subseteq\mathbb{R}$ and $\phi\in\Phi$, define $\phi(X)\equiv\{\phi(x):x\in X\}$. For a profile $S\in\Sigma^N$ define $\phi S \equiv (\phi(S_1),\dots,\phi(S_n))$.
 Note that because $\phi$ is a strictly increasing bijection of the real line onto itself, $X\in\Sigma$ implies that $\phi(X)=(\phi(\inf X),\phi(\sup X))\in\Sigma$.

\begin{description} \item [Ordinal neutrality:] For every $\phi\in\Phi$ and $S\in\Sigma^N$, $\phi(\SO)=f(\phi S)$. \end{description} 

Direction neutrality requires the two directions to be treated alike. 
For $X\in\Sigma$ and $S\in\Sigma^N$, I define $-X\equiv(-\sup X,-\inf X)$ and $-S\equiv(-S_1,\hdots,-S_n)$.

\begin{description} \item [Direction neutrality:] For every $S\in\Sigma^N$, $\SO=-f(-S)$.  \end{description} 

I introduce a class of aggregation rules, called endpoint rules, which I believe have not yet been described in the literature. 
This class of rules is parametrized by two positive integers, $p$ and $q$, such that $p+q\leq n+1$. 
For a profile of standards $S$ and a point $x\in\mathbb{R}$, let $\Gp{S}{x} = \{ i \in N:  \inf S_i < x \}$ be the set of individuals who have $x$, or a point below $x$, in their standard, and let $\Gq{S}{x} = \{ i \in N: \sup S_i > x \}$ be the set of individuals who have $x$, or a point above $x$, in their standard.
An endpoint rule is of the form $$f^{p,q}(S) \equiv \{ x \in\mathbb{R} : |\Gp{S}{x}| \geq p \textnormal{ and } |\Gq{S}{x}| \geq q \}.$$ This rule returns the standard defined by the $p$-th lowest lower endpoint and the $q$-th highest upper endpoint. The restriction  $p+q\leq n+1$ guarantees that the lower endpoint will be to the left of the upper endpoint, and therefore that $f^{p,q}$ is well-defined. 
To see this, note that there are at least $n-p+1$ lower endpoints greater than or equal to the $p$-th lowest lower endpoint, and that this implies that there are at least $n-p+1$ upper endpoints greater than the $p$-th lowest lower endpoint. 
Consequently, the rule is well-defined whenever $q \leq n-p+1$, or when $p + q \leq n + 1$.\footnote{The case $p=q=n$ violates the restriction  $p+q\leq n+1$ for $n>1$. This would result in a standard only when the intersection of all individual standards is non-empty; hence it is not a well-defined rule. One can construct a rule that coincides with $p=q=n$ when this intersection is non-empty; such a rule would violate the responsiveness axiom.}

An example with three individuals is depicted in Figure~\ref{S}, with standards $S_1=(2,4)$, $S_2=(3,6)$, and $S_3=(1,5)$. Here,  $f^{1,1}(S) = (1,6)$ (Figure~\ref{f11S}), $f^{1,3}(S) = (1,4)$ (Figure~\ref{f13S}), and  $f^{2,2}(S) = (2,5)$ (Figure~\ref{f22S}).

\begin{figure}[htb] 
\centering
\subfigure[$S_1=(2,4), \: S_2=(3,6), \: S_3=(1,5)$\label{S}]{
\begin{tikzpicture}[scale=1]
\useasboundingbox (0,1) grid (7,5);
\draw[ultra thick] (1,3) -- (5,3);
\draw[ultra thick] (3,3.5) -- (6,3.5);
\draw[ultra thick] (2,4) -- (4,4);

\draw[<->] (0.5,2.375) -- (6.5,2.375);

\draw (1,2.5)--(1,2.25);
\draw (2,2.5)--(2,2.25);
\draw (3,2.5)--(3,2.25);
\draw (4,2.5)--(4,2.25);
\draw (5,2.5)--(5,2.25);
\draw (6,2.5)--(6,2.25);

\draw (1,2.35) node[anchor=north] {\footnotesize 1};
\draw (2,2.35) node[anchor=north] {\footnotesize 2};
\draw (3,2.35) node[anchor=north] {\footnotesize 3};
\draw (4,2.35) node[anchor=north] {\footnotesize 4};
\draw (5,2.35) node[anchor=north] {\footnotesize 5};
\draw (6,2.35) node[anchor=north] {\footnotesize 6};

\end{tikzpicture}}
\subfigure[$f^{1,1}(S)=(1,6)$\label{f11S}]{
\begin{tikzpicture}[scale=1]
\useasboundingbox (0,1) grid (7,5);
\draw[ultra thick] (1,3) -- (5,3);
\draw[ultra thick] (3,3.5) -- (6,3.5);
\draw[ultra thick] (2,4) -- (4,4);

\draw[<->] (0.5,2.375) -- (6.5,2.375);

\draw[ultra thick] (1,1.75) -- (6,1.75);
\draw[dotted] (1,1.75) -- (1,3);
\draw[dotted] (6,1.75) -- (6,3.5);
\end{tikzpicture}}
\subfigure[$f^{1,3}(S)=(1,4)$\label{f13S}]{
\begin{tikzpicture}[scale=1]
\useasboundingbox (0,1) grid (7,5);
\draw[ultra thick] (1,3) -- (5,3);
\draw[ultra thick] (3,3.5) -- (6,3.5);
\draw[ultra thick] (2,4) -- (4,4);

\draw[<->] (0.5,2.375) -- (6.5,2.375);

\draw[ultra thick] (1,1.75) -- (4,1.75);
\draw[dotted] (1,1.75) -- (1,3);
\draw[dotted] (4,1.75) -- (4,4);
\end{tikzpicture}}
\subfigure[$f^{2,2}(S)=(2,5)$\label{f22S}]{
\begin{tikzpicture}[scale=1]
\useasboundingbox (0,1) grid (7,5);
\draw[ultra thick] (1,3) -- (5,3);
\draw[ultra thick] (3,3.5) -- (6,3.5);
\draw[ultra thick] (2,4) -- (4,4);

\draw[<->] (0.5,2.375) -- (6.5,2.375);

\draw[ultra thick] (2,1.75) -- (5,1.75);
\draw[dotted] (2,1.75) -- (2,4);
\draw[dotted] (5,1.75) -- (5,3);

\end{tikzpicture}}
\caption{Endpoint Rules\label{endpointrules}}
\end{figure}
	
An important subclass of rules is that of the symmetric endpoint rules, where $p = q$. The symmetric endpoint rule with $p=q=\lfloor\frac{n+1}{2}\rfloor$ is the median rule.

I now present the main result.

\begin{thm} \label{mainthm} An aggregation rule $f$ satisfies responsiveness, anonymity, continuity, and ordinal neutrality if and only if it is an endpoint rule. Furthermore, $f$ additionally satisfies direction neutrality if and only if it is a symmetric endpoint rule.\end{thm}

The proof is in the appendix. The five axioms are independent when $n\geq 3$; however, symmetric endpoint rules can be characterized without continuity in the case where $n=2$.\footnote{The proofs of these facts are left as an exercise for the reader.}

The model defines standards as bounded open non-empty intervals of the real line. I note three properties that are important for the result: the richness of the real line, the boundedness of the intervals, and the continuity axiom. 

First, to understand the importance of the richness of the real line, consider an alternative setting in which standards are modeled as bounded non-empty convex subsets of the set of integers. 
As any strictly increasing bijection of the set of integers onto itself is a translation, the ordinal neutrality axiom would become equivalent to translation equivariance. 
It would thus become possible to define an aggregation rule where $\SO = \left[\left\lfloor\frac{1}{n}\sum_i\min(S_i)\right\rfloor,\left\lceil\frac{1}{n}\sum_i\max(S_i)\right\rceil\right]$; one can observe that this rule (a) would satisfy the analogues of the five axioms and (b) is not an endpoint rule.

Second, to understand the importance of boundedness, consider an alternative setting where $\Sigma$ contains all open non-empty intervals of the real line, including half-bounded intervals and the real line itself. 
Note that for every bijection $\phi$ of $\mathbb{R}$ onto itself, $\mathbb{R}=\phi(\mathbb{R})$. 
It follows that the rule defined by $\SO=\mathbb{R}$ for all $S\in\Sigma^N$ (a) would satisfy the analogues of the five axioms and (b) is not an endpoint rule.

Third, to understand the importance of continuity, consider the case where $n\geq 3$ and let $\Sigma^{N\circ}\subseteq\Sigma^N$ be the set of profiles $S$ such that $S_i\cap S_j\neq\varnothing$ for all $i,j\in N$.
One can then define the rule $f$ such that $\SO=f^{1,1}(S)$ for $S\in\Sigma^{N\circ}$ and $\SO=f^{2,2}(S)$ for $S\not\in\Sigma^{N\circ}$. 
One can observe that  this rule (a) satisfies anonymity, responsiveness, ordinal neutrality, and direction neutrality, (b) fails continuity, and (c) is not an endpoint rule.

\section{Strategyproofness\label{sec:strat}}

To study the question of strategyproofness it is necessary to make an assumption about preferences.\footnote{No assumption about preferences has been made up to this point, as judgments need not come from a preference ordering.} \citet{vb:2010} and \citet{fc:2011} define a class of single-peaked preferences on intervals that relies on a concept of betweenness. An interval is defined to be between two other intervals if its lower endpoint is between the lower endpoints of the other two, and if its upper endpoint is between the upper endpoints of the other two. 
For intervals $R,T\in\Sigma$, let $\mathcal{B}(R,T)$ be the set of all $S\in\Sigma$ such that
$$\inf R \leq\inf S \leq\inf T  \textbf{ or } \inf R \geq\inf S \geq\inf T$$ and
$$\sup R \leq\sup S \leq\sup T  \textbf{ or } \sup R \geq\sup S \geq\sup T.$$

Single-peaked preferences are preferences for which (a) there is a unique preferred interval (the `peak') and (b) any interval that is between the peak and a third interval is weakly preferred to that third interval. 
Let $\mathcal{P}$ be the set of complete, reflexive, and transitive preference relations ${\succeq_i}$ on $\Sigma$ for which (a) there exists $S_i^*\in\Sigma$ such that, for all $T\in\Sigma$, $T\succeq_i S_i^*$ implies that $T=S_i^*$, and (b) for $R,T\in\Sigma$, $R\in\mathcal{B}(S_i^*,T)$ implies that $S_i^* \succeq_i R \succeq_i T$. 
For an agent $i\in N$ with preferences ${\succeq_i}\in\mathcal{P}$, I denote by $S_i^*$ that agent's `peak' or most preferred interval.

An aggregation rule is strategyproof if it is always in an agent's interest to report her most preferred interval $S_i^*$, holding constant the reports of the other agents. 
Here, $S$ is a profile of reported standards; it may differ from the agents' ``true'' (or preferred) standards.

\begin{description} \item[Strategyproofness:] For every $S\in\Sigma^N$, $i\in N$, and ${\succeq_i}\in\mathcal{P}$, $f(S_i^*,S_{-i}) \succeq_i \SO$.\end{description} 

This axiom can also be stated without reference to preferences. Consider an individual who changes her judgment from $S_i$ to $S'_i$. The \textit{out-between-ness} axiom requires the outcome $f(S'_i,S_{-i})$ to be further away from $S_i$ than $\SO$; in other words, $\SO$ should be in between $S_i$ and $f(S'_i,S_{-i})$.

\begin{description} \item[Out-between-ness:]  For every $S\in\Sigma^N$, $i\in N$, and $S'_i\in\Sigma$,  $\SO\in\mathcal{B}(S_i,f(S'_i,S_{-i}))$.\end{description}

\begin{lem} \label{Blem} Strategyproofness and out-between-ness are equivalent.\end{lem}

\citet{vb:2010} and \citet{fc:2011} show that the median rule is strategyproof. It is straightforward to show that endpoint rules are strategyproof as well.\footnote{In a finite setting, strategyproofness of endpoint rules can be proven using \citet[Theorem 3]{np:2006}.}

I provide a complete characterization of anonymous and strategyproof rules.

Endpoint rules aggregate the upper and lower endpoints independently; that is, these rules aggregate the lower endpoints without considering the upper endpoints, and vice versa. The property can be defined formally as follows:

\begin{description} \item[Independent aggregation of endpoints:] For every $S,T\in\Sigma^N$, (a) $\inf S_i =\inf T_i $ for all $i\in N$ implies that $\inf \SO = \inf f(T)$ and (b) $\sup S_i =\sup T_i $ for all $i\in N$ implies that $\sup \SO =\sup f(T)$.\end{description}

I show that all strategyproof rules have this property.

\begin{prop}\label{spind}An aggregation rule $f$ satisfies strategyproofness only if it satisfies independent aggregation of endpoints.\end{prop}

The proof of this proposition is in the appendix. 

A consequence of Proposition~\ref{spind} is that, under the assumption of strategyproofness, the aggregation of lower (and upper) endpoints is essentially equivalent to the aggregation of single-peaked preferences on a single issue dimension, as studied in \citet{hm:1980}. 
This allows for an analogue of Moulin's characterization of ``phantom voters'' that includes the endpoint rules as a special case.

Define $\extR\equiv\mathbb{R}\cup\{-\infty,\infty\}$ as the extended reals, with the usual order ($-\infty<x<\infty$ for $x\in\mathbb{R}$), and define $\extlt$ as a binary relation on $\extR$ such that, for $x,y\in\extR$, $x\extlt y$ if (i) $x,y\in\mathbb{R}$ and $x<y$, (ii) $x=-\infty$, or (iii) $y=\infty$.
Let $\extSigma=\{(x,y)\in\extR\times\extR:x\extlt y\}$ be the set of open non-empty intervals of the extended reals, together with the degenerate intervals $(-\infty,-\infty)$ and $(\infty,\infty)$.\footnote{For all $a,b\in\extR$ such that $a\extlt b$, I define $(a,b) = \{x\in\extR : a\extlt x\extlt b\}$. This implies that $(-\infty,-\infty) = \{-\infty\}$, $(\infty,\infty) = \{\infty\}$, and $(-\infty,\infty) = \extR$.}
For $Q=(x,y)\in\extSigma$, define $\inf Q=x$ and $\sup Q=y$.
Note that $\Sigma\subset\extSigma$. 
Let $\med:\extSigma^{2n+1}\rightarrow\extSigma\cup\{\varnothing\}$ be the median rule applied to $2n+1$ intervals of the extended reals; that is, for a vector $Q\in\extSigma^{2n+1}$, $$\med(Q)=\{x\in\mathbb{R}: |\{i\leq 2n+1: \inf Q_i < x\}|,|\{ i\leq 2n+1: \sup Q_i > x\}| \geq n+1\}.$$ 
Note that if $\med(Q)\neq\varnothing$ then $\inf\med(Q)$ and $\sup\med(Q)$ are the median elements of $(\inf Q_1,\dots{},\inf Q_{2n+1})$ and  $(\sup Q_1,\dots{},\sup Q_{2n+1})$, respectively.
 For a profile $S\in\Sigma^N$ and a vector $P\in\extSigma^{n+1}$, define $S\oplus P\in\extSigma^{2n+1}$ as the vector such that $(S\oplus P)_i = S_i$ for $i\leq n$ and $(S\oplus P)_i = P_{i-n}$ for $i> n$. 

The following claim characterizes the class of anonymous and strategyproof rules. The proof relies on \citet{hm:1980} and is found in the appendix.

\begin{claim} \label{claim1} An aggregation rule $f$ satisfies anonymity and strategyproofness if and only if there exists $P\in\extSigma^{n+1}$ such that $\SO = \med(S\oplus P)$ for all $S\in\Sigma^N$.\end{claim}

Here, the components $P_i$ of $P$ may be thought of as ``phantom intervals''; unlike the agents' intervals, however, these include the half-bounded intervals of the form $(-\infty,x)$ and $(x,\infty)$ for $x\in\mathbb{R}$, the fully unbounded interval $(-\infty,\infty)$, and the degenerate intervals $(-\infty,-\infty)$ and $(\infty,\infty)$.\footnote{Alternatively, one may think of these as ``calibration intervals,'' following \citet{wt:2018}.} 

The ordinal neutrality axiom would further imply that the resulting rules are endpoint rules.
This is because ordinal neutrality would eliminate the bounded and half-bounded intervals; thus all phantom intervals would be fully unbounded or degenerate.  
The connection can be easily observed: the $p,q$-{th} endpoint rule corresponds to $p\geq 1$ phantoms of the form $(\infty,\infty)$, $q\geq 1$ phantoms of the form $(-\infty,-\infty)$, and $n+1-p-q$ phantoms of the form $(-\infty,\infty)$.\footnote{In addition, direction neutrality would imply that the number of degenerate intervals $(-\infty,-\infty)$ must equal the number of degenerate intervals $(\infty,\infty)$, and therefore that $p=q$.}

However, endpoint rules can be characterized without the full strength of the ordinal neutrality axiom. A much weaker axiom, translation equivariance, is sufficient to reach the same result.
For $X\subseteq\mathbb{R}$ and $b\in\mathbb{R}$, define $X+b = \{x+b:x\in X\}$. For $S\in\Sigma^N$, define $S+b=(S_1+b,\dots{},S_n+b)$.

\begin{description} \item[Translation equivariance:] For every $S\in\Sigma^N$ and $b\in\mathbb{R}$, $f(S+b) = \SO+b$.\end{description}

The endpoint rules are characterized by anonymity, strategyproofness, and translation equivariance.\footnote{Translation equivariance is presented as a simple example of an axiom that eliminates bounded and half-bounded intervals; however, it is not the only such axiom. A strict version of the responsiveness axiom would also characterize the endpoint rules when combined with anonymity and strategyproofness: for all $S,T\in\Sigma^N$, if $\inf T_i<\inf S_i$ and $\sup T_i>\sup S_i$ for all $i\in N$, then $\inf f(T)<\inf \SO$ and $\sup f(T)>\sup \SO$. I thank Jianrong Tian for this observation.}

\begin{thm} \label{thm:tisp} An aggregation rule $f$ satisfies anonymity, strategyproofness, and translation equivariance if and only if it is an endpoint rule. \end{thm}

The proof is in the appendix.

\section{Other applications \label{oa}}

The model is motivated by legal rules according to which actions are judged by reference to community standards of behavior. 
However, the aggregation of intervals is relevant elsewhere. 
To varying extents, the results reported in this paper may be applicable to these other contexts.

\subsection{Types of applications \label{ta}}

A first type of application involves community standards that arise outside of law.
Computer scientists seek to align artificial intelligence systems with human values \citep{ig:2020}. 
Community standards may guide this alignment in a manner analogous to their use in aligning the law.\footnote{\citet{sr:2024} explicitly proposes that the common law reasonable person standard be used to guide the alignment of artificial intelligence.}

As in the legal context, one may envision a community standard as an aggregate of individual standards.\footnote{A significant subset of the alignment literature recognizes the normative difficulty of aggregating human values \citep[see][]{ig:2020,am:2023,cfh:2024,ghm:2024}.}
In some applications it may be helpful to model these standards as bounded open non-empty intervals of the real line.
For example, \citet{wgh:2025} evaluates the outputs of large language models on a left--right dimension \citep[see also][]{kes:2025}.

A second type of application involves a panel of experts that needs to provide an aggregate opinion, where the expert opinions take the form of a range. 
\citet{fc:2011} provides an example of climatologists who must aggregate their views about the likely range of global warming outcomes, when distrust hampers their ability to report anything more than a range of possibilities. 
These expert opinions may be in the form of a range because the estimate can be limited by information constraints, and because the outcome can depend on choices made by future policy makers. 
\citep[See][for additional reasons why experts may make interval forecasts.]{cc:1993}
The opinions may also represent normative uncertainty. 
Judges may agree that the law is unclear within a certain range, but not agree on the range.
Whether a tortfeasor has committed a clear violation of the law may be relevant, for example, in the determination of punitive damages.

A third type of application involves a group of principals that must delegate a decision to a better informed agent. 
A panel of judges may need to agree on a sentence that takes the form of a range (e.g., five to ten years) so that a parole board can base the ultimate sentence on the convict's subsequent behavior. 
A legislature may wish to choose a sentencing range for certain crimes, so as to constrain future judges \citep{fc:2011}. 
A significant literature in political science addresses the problem of bureaucratic drift, or how administrative agency decisions can deviate from those preferred by the legislature \citep[see][]{eo:1994,eo:1996,eo:1999,eo:2008,hs:2002,hs:2006,sg:2009,gp:2012,ck:2014}. 
This literature primarily follows the model of delegation found in the principal--agent literature \citep[see][]{bh:1977,bh:1984,ns:1991,am:2008res,ak:2022}, in which delegation commonly takes the form of an interval of permissible alternatives.

A difficulty in this last literature is that, as legislatures are not unitary entities, the identification of the legislature with a preference is problematic for reasons first identified by \citet{ar:1951}.
The papers in this literature tend to assume away the difficulty of combining legislators' preferences by making strong assumptions; for example, the existence of a median legislator \citep[see, e.g.,][]{eo:1994}. 
Having done so, these papers can first aggregate the preferences (to that of the median legislator) and then find the legislature's preferred interval.
However, rather than make strong assumptions on the legislators' preferences, one could reverse the order of these steps---first finding each individual legislator's preferred interval and then aggregating these to find the legislature's interval.

A fourth type of application involves a group of countries trying to determine the terms of a multilateral treaty. 
A trade agreement may place bounds on the level of permissible tariffs \citep[see][]{ak:2022}. 
A defense treaty may place limits on the level of military spending. 
Members of the military pact may want their fellow members to commit to spending a certain amount of money each year (to avoid free riding) but may also want to limit the maximum amount that their partners spend (because of distrust).

\subsection{Relevance of the model\label{rm}}

The primitive of the model is a profile of intervals, one for each individual, drawn from the set of all bounded open non-empty intervals of the real line.
The output is an interval drawn from this same set.
These intervals represent judgments, at both the individual and the community level.

This primitive is justified because community standards are based on judgments, and not on preferences. 
Courts judge the permissibility of actions according to whether the actions are considered acceptable in the community, and not according to whether members of the community would like them to be allowed. 
The individual judgments are not necessarily the result of preferences over alternatives, or even of preferences over standards as described in Section~\ref{sec:strat}. 
An individual may judge an act to be permissible despite a preference that it be banned, or vice versa.

To what extent does the model apply in contexts other than law?
Artificial intelligence resembles law with respect to community standards.
If the goal is to align the behavior of an artificial intelligence system with human values, it is reasonable to start with human judgments about the acceptability of behavior.
Some researchers have suggested as much; for example, in the context of supervised machine learning, \citet{glp:2022} proposes to use a jury of evaluators whose task is to judge whether content is toxic.
However, as artificial intelligence is a broad field, other primitives may also be appropriate.

In some contexts, judgments may arise as a predictable result of underlying preferences and information.
For example, a panel of experts that needs to provide a collective judgment is primarily tasked with aggregating information.
But their preferences may also differ because of heterogeneity in their attitudes toward risk or in their discount rates.
Judges, legislators, and governments similarly differ in their preferences and information.
For this reason, one might start with a primitive of preferences and information, rather than the resultant judgments.

However, it is known to be difficult to aggregate preferences \citep[see][]{ar:1951}, and even more difficult to aggregate preferences and information simultaneously \citep[see][]{hz:1979}.
Furthermore, it may be hard to incentivize the agents to report their preferences and information truthfully.
For this reason, \citet{fc:2011} suggests that distrust may lead agents to limit their reporting to an interval.
Aggregating the judgments directly will not solve all problems, but may in some cases simplify the task.

Nonetheless, a model of preferences and information can guide the analysis even when aggregation is impractical.
If intervals come from preferences or information, it is possible that some intervals will never arise in practice.
This would challenge the full domain assumption---that individual standards may be any bounded open non-empty interval of the real line.
Endpoint rules may not uniquely satisfy the axioms on a restricted domain. 

\subsection{Relevance of the axioms\label{ra}}

To what extent are the axioms meaningful when applied to other contexts?
I discuss the relevance of the axioms---anonymity, responsiveness, continuity, ordinal neutrality, direction neutrality, translation equivariance, and strategyproofness---to these alternative applications.

\subsubsection{Anonymity}

Anonymity is important in the context of community standards. 
To the extent that individuals are members of the relevant community, their views are generally treated equally.\footnote{Of course, judges and juries generally base their determinations on their understandings of the relevant community, and do not explicitly calculate the community standard on the basis of publicly available data. As a consequence, one cannot know whether they treat these individuals' views equally in practice. Nonetheless, the language generally used to justify these standards does not provide a reason to distinguish between community members.}
Similarly, members of legislatures, judges on a court, and independent states are generally presumed to be equal.

Nonetheless, there are at least two reasons why anonymity may at times be undesirable.
First, the judgments may come from individuals with different levels of expertise.
One might want to include a graduate student on an expert panel but nonetheless treat that student differently from a Nobel laureate.\footnote{\citet{mbk:2026} notes that, in the context of artificial intelligence, ``a question about medical ethics might warrant consulting both a general population and healthcare professionals'' but not necessarily in equal degrees.}
This may also apply in select cases of community standards, such as a ``reasonable doctor'' standard that relies on the doctors' expertise.\footnote{The reasonable doctor standard may be consistent with the anonymity axiom if membership in the community is restricted to doctors, but would not be consistent if the judgment of an experienced doctor is given a weight equal to that of three medical residents.}

Second, some parties may have unequal bargaining power or unequal legal rights.
For example, consider the case of treaty negotiations. 
In principle, all states are equal.\footnote{``The Organization is based on the principle of the sovereign equality of all its Members.'' U.N. Charter art. 2, para. 1.}
Nonetheless, this principle has exceptions,\footnote{The U.N. Security Council is a notable case, as it has five permanent members (United States, United Kingdom, France, Russia, and China) with a veto over Security Council decisions. U.N. Charter art. 23, para. 1; art. 27, para. 3. However, the Security Council does not negotiate treaties.} and even countries with equal rights may agree to a non-anonymous rule to secure the participation of a more powerful state in treaty negotiations. 

\subsubsection{Responsiveness}

The responsiveness axiom is important in the context of community standards. 
No individual should be punished because of a widening of the others' bounds of tolerance.

The value of this axiom will otherwise vary across cases. 
It will generally be desirable if the intervals arise from preferences.
However, this is less clear if the intervals arise from information.
An expert who provides too wide a range may be considered to be useless, and ignored \citep[cf.][which weights interval judgments by their inverse width]{iy:1997}.
Alternatively, an expert who provides a wide range may be providing valuable information---that in the expert's eyes, very little is known---and if so, a responsive aggregation rule would respect that information, either by widening the range of tolerance, or by keeping it exactly the same.

\subsubsection{Continuity}

The continuity axiom is important in the context of community standards. 
It requires, essentially, that two similar communities should have similar community standards.
A community standard that changes drastically because of an imperceptible change in individual standards would lead to serious concerns about arbitrariness and the rule of law.

Continuity is less defensible in contexts in which simplicity may be more important than accuracy.  
Consider the example of sentencing ranges. 
While judges may be able to make fine-grained assessments as to the appropriate punishment of a particular offender, it can be administratively difficult to offer finely tailored punishments.
For this reason, one might design an aggregation rule whose possible outputs form a discrete set; for example, a rule that only produces ranges whose endpoints are whole numbers of months.
Such a rule cannot satisfy continuity unless it is constant.
If the rule operates by merely rounding the output of a continuous rule, the discontinuity will be limited to one month.

\subsubsection{Ordinal Neutrality}

The ordinal neutrality axiom requires the aggregation rule to disregard the cardinal properties of the real line.
This suggests that ordinal neutrality is desirable when aggregating qualitative, but not quantitative, standards.
For example, a community standard of offensiveness is likely to be qualitative---one may be able to order works in terms of offensiveness---but there may be no objective way to make cardinal comparisons. 
Similarly, it is hard to quantify distance on a left--right political dimension. 
In contrast, a range of temperatures or probabilities is quantitative.

In some contexts it can be hard to determine whether standards are qualitative or quantitative.
A court may need to determine whether a particular speed was reasonable.
But because there is no objective cardinal mapping between the speeds and the extent to which they are reasonable, that court may choose to disregard the cardinal relationship between the possible speeds.

Ordinal neutrality is also less desirable in cases in which there are focal points.
In the example of sentencing ranges given above, whole numbers of months are meaningful.
In other contexts, there may be a natural zero or a point that represents the status quo.
Any rule that treats these points as special would violate ordinal neutrality.

\subsubsection{Direction Neutrality}

The direction neutrality axiom restricts the aggregation rule from treating one direction of the real line differently from the other.
This axiom is appropriate when there is no justification for such a distinction. 
In a free society, a court that determines the permissibility of political speech should not favor the political right over the political left, or vice versa.

However, in other contexts, direction can matter. 
For example, while it can be dangerous to drive both at high and at low speeds, the nature of these dangers is different. 
When determining a reasonable speed, one may wish to design a rule that treats judgments of ``unreasonably fast'' differently from how it treats judgments of ``unreasonably slow.''

\subsubsection{Translation Equivariance}

The translation equivariance axiom is applicable in all contexts in which the stronger ordinal neutrality axiom applies.
It may also be applicable when aggregating quantitative standards for which there is no natural zero.
For example, in the context of the climatologists \citep{fc:2011}, it should not matter whether their estimates are provided relative to a 1990 or 2026 baseline.\footnote{Of course, translation equivariance could be less defensible if the range of temperatures contains or approaches a meaningful fixed point, such as absolute zero.} 
Similarly, it may not matter whether a range of dates is measured relative to the birth of Jesus or the introduction of the Model T \citep[see][]{alh:1932}.

Translation equivariance is less defensible in cases in which there is a natural zero.
In the context of treaties, a rule may distinguish between positive tariffs, negative tariffs (import subsidies), and a tariff of zero (duty-free trade).  
In the criminal context, a rule may give special treatment to a zero sentence (in which the defendant stays out of prison).
A rule used to determine a range of probabilities may need to account for the fact that probabilities are bounded by zero and one.

\subsubsection{Strategyproofness}

The strategyproofness axiom is important in cases in which individuals have preferences over intervals, and in which those preferences satisfy the generalized single-peakedness condition introduced by \citet{vb:2010} and \citet{fc:2011} and followed in this paper.
Legislators who delegate discretion to bureaucratic agencies have preferences over the delegated interval.
It is possible that these preferences are generalized single-peaked, though this remains to be shown.

In the context of community standards in law, the strategyproofness axiom is mostly irrelevant, as individuals do not directly report their standards.
In some cases, however, it may be possible for individuals to change their standards with the aim of affecting a future community standard.\footnote{For example, states are often concerned with how their stated standards will, in the long run, affect the creation of customary international law.}

In the context of artificial intelligence, community standards are generally derived from data, such as feedback from human evaluators \citep[see][]{glp:2022}.
Several researchers have recognized that human evaluators may strategically alter their feedback \citep[see][]{slh:2024,kgmk:2025,mpd:2025}.
In a related setting, \citet{kgmk:2025} explicitly proposes the use of a coordinate-wise median to combat manipulability.
 
The ordinal neutrality and strategyproofness axioms exclude the possibility of averaging rules, in which aggregate endpoints are averages of individual endpoints.
Ordinal neutrality is a reasonable justification in the contexts in which averaging does not make much practical sense---for example, when taking a numerical average would not be particularly meaningful.
However, even in contexts in which averaging makes sense on its face, strategyproofness implies that it should not be used in cases where the parties are interested in the outcome and have some ability to misrepresent their information or preferences.

The other axioms used in this paper are compatible with averaging. Define the simple averaging rule by the equally weighted average of the endpoints.

\begin{description}\item[Simple averaging rule:] 
For all $S\in\Sigma^N$, $f(S)=(\frac{1}{|N|}\sum_{i\in N} \inf S_i, \frac{1}{|N|}\sum_{i\in N} \sup S_i)$.
\end{description} 

The simple averaging rule satisfies all axioms used in Theorem~\ref{mainthm} except for ordinal neutrality, and all axioms used in Theorem~\ref{thm:tisp} except for strategyproofness.

\section{Conclusion}

I have introduced the endpoint rules and have shown that they are characterized by responsiveness, anonymity, continuity, and ordinal neutrality in this setting. 
The addition of  direction neutrality yields a characterization of the subfamily of symmetric endpoint rules.
I have shown that with a suitable restriction on preferences, endpoint rules are strategyproof, and that all strategyproof, anonymous, and translation equivariant aggregation rules are endpoint rules.

Section~\ref{oa} suggests that the results in this paper may be applicable in contexts outside of community standards. 
However, it is important to keep in mind the ways in which the legal problem is distinct. 
Some of these differences have been described in the paper.
For example, while the concept of preference is generally irrelevant to community standards in law, it may be important in other settings.
 
I will highlight three important distinctions here.
First, individual standards are subjective in most legal applications.
They can be motivated by abstract concerns, or no concerns at all. 
Second, there is no objectively correct answer.
An aggregation rule is simply a systematic method of combining subjective judgments.\footnote{To be precise, individual standards cannot be objectively judged for correctness. A trier of fact who determines the community standard on the basis of their familiarity with the community can, of course, have incorrect perceptions of what individuals in the community believe. In this sense, the community standard is a constructed fact, rather than a ``real'' or objective one \citep[see][Section 9.1]{dl:2026}.}
Third, community standards do not necessarily serve consequentialist ends.
In the context of obscenity, offense to community standards is often the justification (and not merely the test) for criminal prosecution \citep[see][]{pd:1965}.
In contrast, the reasonable person standard may exist to reduce the cost of accidents, though this is not universally accepted as a goal.

These three distinctions may not hold in settings outside of community standards.
For example, interval judgments may have an objective component in the problem of forecasting described by \citet{fc:2011}. 
This is because the scientists may found their judgments on the same underlying base of evidence.
There also may be a correct answer, and even if not, a forecast can be compared to the realization of an event.
The existence of \textit{ex post} evidence makes it possible to evaluate aggregation methods in terms of the quality of their results \citep[see, e.g.,][]{iy:1997,gtw:2017}.
And there may be a clear goal: climate forecasts can be used to allocate resources now to forestall a future disaster.

An understanding of the contexts in which intervals are used may facilitate the design of aggregation rules.
For example, judges (or politicians) may agree that the correct sentencing range is the one that optimizes deterrence.
Even if the range cannot be verified, one may  design an aggregation rule to maximize expected welfare or the likelihood of getting the correct answer \citep[cf.][]{hpy:1988}.

I leave to future research the question of whether more general results can be established by focusing on the abstract properties of the betweenness relation and the order. 
To provide a short example: consider the case of a decision that must be made on two (or more) dimensions. 
One must decide not only how fast it is reasonable to drive, but also, how much training drivers should have before getting behind the wheel. 
There may be a tradeoff; at higher speeds, more training is necessary, although different individuals may have different views about the right tradeoff. 
This is a much more complex question, and it is difficult to study. 
In the case of the real line, the concept of betweenness implies intervals, which can be identified with points in two-dimensional space (such that $x_1 < x_2$). 
In the case of multidimensional space, however, betweenness simply implies convexity, and there is no similarly easy way to represent these convex sets. 
In addition, there may be interesting problems with different underlying structures, for which the simple assumption of Euclidean space may not be applicable.

\section*{Appendix}
I first state and prove the following lemma.

\begin{lem}\label{lem:pq} An endpoint rule is a symmetric endpoint rule if and only if it satisfies direction neutrality.\end{lem}

\begin{proof}[Proof of Lemma~\ref{lem:pq}] That symmetric endpoint rules satisfy direction neutrality is trivial. Let $f^{p,q}$ be an endpoint rule that satisfies direction neutrality. I show that $p=q$. 

Let $S\in\Sigma^N$.
By the definition of an endpoint rule, $f^{p,q}(S)=\{x: |\Gp{S}{x}| \geq p \textnormal{ and } |\Gq{S}{x}| \geq q \}$. 
Because $\inf(-S_i)=-\sup S_i$, it follows that $\Gp{-S}{x}=\{i:\inf(-S_i)<x\}= \{i:-\sup S_i < x\}=\{i:\sup S_i > -x\}=\Gq{S}{-x}$. 
By a similar argument, one can establish that $\Gq{-S}{x}=\Gp{S}{-x}$. 
It follows that $-f^{p,q}(-S)=\{-x: |\Gp{-S}{x}| \geq p \textnormal{ and } |\Gq{-S}{x}| \geq q \}=\{-x: |\Gq{S}{-x}| \geq p \textnormal{ and } |\Gp{S}{-x}| \geq q \}=\{x: |\Gq{S}{x}| \geq p \textnormal{ and } |\Gp{S}{x}| \geq q \}=f^{q,p}(S)$.
By direction neutrality, $f^{p,q}(S)=-f^{p,q}(-S)$; therefore, $f^{p,q}(S)=f^{q,p}(S)$.

Let $S^n\in\Sigma^N$ be such that $S^n_i=(2i-1,2i)$. 
Then $\inf f^{p,q}(S^n)=2{p}-1$ and $\inf f^{q,p}(S^n)=2{q}-1$. 
Because $f^{p,q}(S^n)=f^{q,p}(S^n)$, it follows that $2{p}-1=2q-1$ and therefore that $p=q$.
\end{proof}

\begin{proof}[Proof of Theorem~\ref{mainthm}] It is straightforward to show that endpoint rules satisfy responsiveness, anonymity, continuity, and ordinal neutrality, and that symmetric endpoint rules satisfy direction neutrality. 
Let $f$ satisfy responsiveness, anonymity, continuity, and ordinal neutrality. 
I show that $f$ must be an endpoint rule. 
Because, by Lemma~\ref{lem:pq}, an endpoint rule is a symmetric endpoint rule if and only if it satisfies direction neutrality, this is sufficient to complete the proof.

For $S\in\Sigma^N$ define $L(S)\equiv\{\inf S_i : i\in N\}$ and $U(S)\equiv\{\sup S_i : i\in N\}$. 
Because $\tau$ has a base consisting of the sets $\{(x',y'):x<x'<x'',\ y<y'<y''\}$, the maps $(a,b)\mapsto a$ and $(a,b)\mapsto b$ are continuous on $(\Sigma,\tau)$; consequently $S\mapsto\inf S_i$, $S\mapsto\sup S_i$, and, by continuity of $f$, $S\mapsto\inf\SO$ and $S\mapsto\sup\SO$ are continuous on $(\Sigma^N,\tauprod)$.

\textbf{Step 1:} For all $S\in\Sigma^N$, $\inf\SO\in L(S)$ and $\sup\SO\in U(S)$.

Let $S\in\Sigma^N$ and suppose by means of contradiction that $\inf \SO \not\in L(S)$. 
Let $\phi\in\Phi$ be such that $\phi(\inf\SO)>\inf\SO$ and $\phi(x)=x$ for all $x\in (L(S)\cup U(S))\setminus\{\inf\SO\}$; such a $\phi$ exists because $L(S)\cup U(S)$ is finite. 
Because $\inf\SO\not\in L(S)$, $\inf \phi(S_i)=\inf S_i$ for all $i\in N$.
By the definition of $\phi$, $\sup \phi(S_i) \geq \sup S_i$ for all $i\in N$. 
It follows that $S_i\subseteq \phi(S_i)$ for all $i\in N$.
By responsiveness and ordinal neutrality, $\SO\subseteq f(\phi S)=\phi(\SO)$, and thus $\inf\SO\geq\phi(\inf\SO)>\inf\SO$, a contradiction. 
The proof that $\sup\SO\in U(S)$ is symmetric.

\textbf{Step 2:} If $S,S'\in\Sigma^N$ satisfy $\inf S_i=\inf S'_i$ for all $i\in N$, then $\inf\SO=\inf f(S')$.

For $\lambda\in[0,1]$ let $R^\lambda\in\Sigma^N$ be such that $R^\lambda_i=\left(\inf S_i,\,(1-\lambda)\sup S_i+\lambda\sup S'_i\right)$ for all $i\in N$.
Note that  $R^0=S$, $R^1=S'$, and that $L(R^\lambda)=L(S)$ for all $\lambda\in[0,1]$. 
Define $h:[0,1]\rightarrow\mathbb{R}$  such that $h(\lambda)=\inf f(R^\lambda)$.
Note that (a) the function $h$ is continuous on $[0,1]$ (because the endpoints of $R^\lambda$ depend continuously on $\lambda$) and (b) by Step 1 the range of $h$ is contained in $L(S)$.
Because $L(S)$ is finite, it follows from the intermediate value theorem that $h$ must be a constant function.
Therefore,  $\inf\SO=\inf f(R^0)=\inf f(R^1)=\inf f(S')$.

\textbf{Step 3:} Let $T^*\in\Sigma^N$ be such that $T^*_i=(i,i+1)$ for all $i\in N$, and let $\dot{p}\equiv\inf f(T^*)$. By Step 1, $\dot{p}\in N$.
For $S\in\Sigma^N$ and $p\in N$, define $\ell(S,p)\equiv\inf\{x:|\Gp{S}{x}| \geq p\}$. 
Note that $\ell(S,p)$ is the $p$-th lowest of the lower endpoints of $S$.
I show that $\inf\SO=\ell(S,\dot{p})$ for all $S\in\Sigma^N$.

First, define $\Sigma^{N\bullet}\equiv\{S\in\Sigma^N :\inf S_i\neq \inf S_j\textnormal{ for all }i\neq j\}$ and let $S\in\Sigma^{N\bullet}$.
Let $\pi$ be a permutation of $N$ such that $\inf (\pi S)_1 < \cdots < \inf (\pi S)_n$.
Let $\phi\in\Phi$ be such that $\phi(\inf(\pi S)_j) = j$ for all $j\in N$, and note that $\inf(\phi \pi S)_j = \phi(\inf (\pi S)_j) = j =\inf T^*_j$ for all $j\in N$.
By anonymity and Step 2, $\inf f(\phi S)=\inf f(\phi \pi S)=\inf f(T^*)=\dot{p}$. 
By ordinal neutrality, $\inf f(\phi S)=\phi(\inf\SO)$.
It follows that  $\inf\SO=\phi^{-1}(\dot{p})=\inf S_{\pi(\dot{p})}=\ell(S,\dot{p})$.

Note that (a) both $\inf f(S')$ and $\ell(S',\dot{p})$ are continuous functions of $S'$ on $(\Sigma^N,\tauprod)$, (b) $\inf f(S')=\ell(S',\dot{p})$ whenever $S'\in\Sigma^{N\bullet}$ (as shown above), and (c) $\Sigma^{N\bullet}$ is dense in $(\Sigma^N,\tauprod)$.
It follows that the two functions agree everywhere \citep[see][Problem 7C]{sw:2004}.

\textbf{Step 4:} A symmetric argument (that exchanges the roles of lower and upper endpoints) shows that there exists $\dot{q}\in N$ such that, for all $S\in\Sigma^N$, $\sup\SO=u(S,\dot{q})\equiv\sup\{x:|\Gq{S}{x}| \geq \dot{q}\}$, the $\dot{q}$-th highest of the upper endpoints of $S$. 
Because $\SO$ is an open non-empty interval, it follows that for all $S\in\Sigma^N$, $\SO=(\ell(S,\dot{p}),u(S,\dot{q}))$.
It follows that $f(T^*)=(\dot{p},\,n+2-\dot{q})$ and therefore that $\dot{p}<n+2-\dot{q}$.
Because $\dot{p}$ and $\dot{q}$ are integers, $\dot{p}+\dot{q}\leq n+1$. 
Therefore, $f$ is an endpoint rule; that is, $\SO=f^{\dot{p},\dot{q}}(S)$ for all $S\in\Sigma^N$.
\end{proof}

\begin{proof}[Proof of Lemma~\ref{Blem}]
First, let $f$ satisfy out-between-ness  and let $S\in\Sigma^N$, $i\in N$, and ${\succeq_i}\in\mathcal{P}$ with associated peak $S_i^*$. 
By out-between-ness, $f(S_i^*,S_{-i})\in\mathcal{B}(S_i^*,\SO)$.
By the definition of $\mathcal{P}$, it follows that $f(S_i^*,S_{-i}) \succeq_i \SO$.

Next, let $f$ satisfy strategyproofness and suppose, contrariwise, that there exist $S\in\Sigma^N$, $i\in N$, and  $S_i^*\in\Sigma$ such that $f(S_i^*,S_{-i})\not\in\mathcal{B}(S_i^*,\SO)$. For an interval $T_i\in\Sigma$ define $d(T_i) = |\inf S_i^* - \inf T_i| + |\sup S_i^* - \sup T_i|$, and define

$$u(T_i) = \left\{\begin{array}{ll} d(T_i), & \textnormal{for } T_i\in\mathcal{B}(S_i^*,\SO) \\ & \\ d(T_i) + d(\SO),& \textnormal{otherwise.}\end{array}\right.$$

Let $\succeq^*_i$ be the preference such that $T_i \succeq^*_i T'_i$ if and only if $u(T_i) \leq u(T'_i)$.

I first show that ${\succeq^*_i}\in\mathcal{P}$.
It is trivial to see that $S_i^*$ is the peak of $\succeq^*_i$. 
To see that it is single-peaked, let $T_i$ and $T'_i$ be intervals such that $T_i\in\mathcal{B}(S_i^*,T'_i)$.
I show that $T_i \succeq^*_i T'_i$.
Note that (i) $d(T_i) \leq d(T'_i)$ and (ii) $T_i\in\mathcal{B}(S_i^*,\SO)$ if $T'_i\in\mathcal{B}(S_i^*,\SO)$.
It follows that $u(T_i) \leq u(T'_i)$.

Because $f(S_i^*,S_{-i})\not\in\mathcal{B}(S_i^*,\SO)$, it follows that $f(S_i^*,S_{-i})\neq S_i^*$ and therefore that $d(f(S_i^*,S_{-i})) > 0$.
It follows that $u(f(S_i^*,S_{-i})) = d(f(S_i^*,S_{-i})) + d(\SO) > d(\SO) = u(\SO)$, and therefore that $f(S_i^*,S_{-i}) \not\succeq^*_i \SO$. 
However, by strategyproofness, as ${\succeq^*_i}\in\mathcal{P}$, it follows that $f(S_i^*,S_{-i}) \succeq^*_i \SO$, a contradiction.
\end{proof}

\begin{proof}[Proof of Proposition~\ref{spind}.]
I will show that strategyproofness implies the independent aggregation of the lower endpoints. That strategyproofness implies the independent aggregation of the upper endpoints follows from a dual argument.  

Let $S,T\in\Sigma^N$ be such that $\inf S_i =\inf T_i $ for all $i\in N$ and let $f$ satisfy strategyproofness. I will show that $\inf \SO =\inf f(T) $.
For $j\in N$ let $S^j\in\Sigma^N$ be such that $S^j_i=S_i$ for $i\leq j$ and such that $S^j_i=T_i$ otherwise.
Define $S^0\equiv T$.

Let $k\in N$. It is sufficient to prove that $\inf f(S^{k-1}) = \inf f(S^{k})$.

Because $f$ is strategyproof, it satisfies out-between-ness by Lemma~\ref{Blem}. 
Because $S^{k-1}$ and $S^{k}$ differ only in coordinate $k$, out-between-ness applied to each of them gives $f(S^{k-1})\in\mathcal{B}(S^{k-1}_k,f(S^{k}))$ and $f(S^{k})\in\mathcal{B}(S^{k}_k,f(S^{k-1}))$. 
Because $\inf S^{k-1}_k=\inf S^{k}_k$, this implies $\inf f(S^{k-1})=\inf f(S^{k})$.
\end{proof}

\begin{proof}[Proof of Claim~\ref{claim1}] First, let $P\in\extSigma^{n+1}$ be such that $\SO = \med(S\oplus P)$ for all $S\in\Sigma^N$. That $f$ satisfies anonymity and strategyproofness is straightforward.
Next, let $f$ satisfy anonymity and strategyproofness. 
I show that there exists $P\in\extSigma^{n+1}$ such that $\SO = \med(S\oplus P)$ for all $S\in\Sigma^N$.
For two points $x,y\in\mathbb{R}$, let $B(x,y)=\{z\in\mathbb{R}:x\leq z\leq y\textnormal{ or }y\leq z \leq x\}$.

For a function $g:\mathbb{R}^N\rightarrow\mathbb{R}$, define $g$ to be anonymous if for every permutation $\pi$ of $N$, $g(\mathbf{x}_1,\hdots{},\mathbf{x}_n)=g(\mathbf{x}_{\pi(1)},\hdots{},\mathbf{x}_{\pi(n)})$.
A complete and transitive preference relation $\succeq_i$ over $\mathbb{R}$ is single-peaked with peak $p_i\in\mathbb{R}$ if $p_i\succ_i x$ for all $x\neq p_i$ and, for all $x,y\in\mathbb{R}$, $x\in B(p_i,y)$ implies that $x\succeq_i y$. Define $g$ to be strategyproof if, for every agent $i\in N$, every single-peaked preference relation $\succeq_i$ over $\mathbb{R}$ with peak $p_i$, and every $\mathbf{x}\in\mathbb{R}^N$, $g(p_i,\mathbf{x}_{-i})\succeq_i g(\mathbf{x})$.

By Proposition~\ref{spind}, because $f$ is strategyproof, it satisfies independent aggregation of endpoints. 
Consequently, there exist $\underline{g},\overline{g}:\mathbb{R}^N\rightarrow\mathbb{R}$ such that for all $S\in\Sigma^N$, $\SO = (\underline{g}(\inf S_1,\hdots,\inf S_n),\overline{g}(\sup S_1,\hdots,\sup S_n))$.
I first prove that $\underline{g}$ is anonymous and strategyproof. 
A similar argument shows that $\overline{g}$ is anonymous and strategyproof.

To show that $\underline{g}$ is anonymous, let $\pi$ be a permutation of $N$ and let $S\in\Sigma^N$. Then $\inf f(\pi S)=\underline{g}(\inf (\pi S)_1,\hdots,\inf (\pi S)_n)$. By anonymity, $\SO=f(\pi S)$; therefore, $\underline{g}(\inf S_1,\hdots,\inf S_n)=\underline{g}(\inf (\pi S)_1,\hdots,\inf (\pi S)_n)$. 

To show that $\underline{g}$ is strategyproof, let $i\in N$, let $\succeq_i$ be a single-peaked preference relation over $\mathbb{R}$ with peak $p_i$, and let $\mathbf{x}\in\mathbb{R}^N$.
Let $S\in\Sigma^N$ be such that $S_j=(\mathbf{x}_j,\mathbf{x}_j+1)$ for all $j\in N$, and let $T\in\Sigma^N$ be such that $T_i=(p_i,p_i+1)$ and $T_j=S_j$ for all $j\neq i$.
Then $\underline{g}(\mathbf{x})=\inf \SO$ and $\underline{g}(p_i,\mathbf{x}_{-i})=\inf f(T)$.

Because $f$ is strategyproof, by Lemma~\ref{Blem} it satisfies out-between-ness.
Because $T_j=S_j$ for all $j\neq i$, it follows that $f(T)\in\mathcal{B}(T_i,\SO)$.
Therefore, $\inf f(T)\in B(\inf T_i,\inf \SO)$ and thus $\underline{g}(p_i,\mathbf{x}_{-i})\in B(p_i,\underline{g}(\mathbf{x}))$.
By single-peakedness of $\succeq_i$, $\underline{g}(p_i,\mathbf{x}_{-i})\succeq_i\underline{g}(\mathbf{x})$. Thus $\underline{g}$ is strategyproof.

Because $\underline{g}$ is anonymous and strategyproof, it follows from \citet[Proposition 2]{hm:1980} that there exist $n+1$ elements $\alpha_1,\hdots{},\alpha_{n+1}\in\extR$ such that, for all $\mathbf{x}\in \mathbb{R}^N$, $\underline{g}(\mathbf{x})=m(\mathbf{x}_1,\hdots{},\mathbf{x}_n,\alpha_1,\hdots{},\alpha_{n+1})$, where $m$ is the function that selects the median element of $(\mathbf{x}_1,\hdots{},\mathbf{x}_n,\alpha_1,\hdots{},\alpha_{n+1})$ with respect to $\leq$.
A similar argument shows that there exist $n+1$ elements $\beta_1,\hdots{},\beta_{n+1}\in\extR$ such that, for all $\mathbf{x}\in \mathbb{R}^N$, $\overline{g}(\mathbf{x})=m(\mathbf{x}_1,\hdots{},\mathbf{x}_n,\beta_1,\hdots{},\beta_{n+1})$.

To finish the proof, it is sufficient to show that there is a permutation $\pi$ of $\{1,\hdots{},n+1\}$ such that $\alpha_i\extlt\beta_{\pi{(i)}}$ for all $i\in\{1,\hdots{},n+1\}$; in such a case define $P_i = (\alpha_i,\beta_{\pi{(i)}})$. 

Suppose by means of contradiction that no such permutation $\pi$ exists. 
Order the $\alpha$s so that $\alpha_{(1)}\leq\cdots\leq\alpha_{(n+1)}$ and the $\beta$s so that $\beta_{(1)}\leq\cdots\leq\beta_{(n+1)}$. 
By the supposition, there exists $k\leq n+1$ such that $\alpha_{(k)}\extlt\beta_{(k)}$ is false; that is, $\alpha_{(k)}\neq-\infty$, $\beta_{(k)}\neq\infty$, and $\beta_{(k)}\leq\alpha_{(k)}$. Let $t\in\mathbb{R}$ be such that $\beta_{(k)}\leq t\leq\alpha_{(k)}$.
Note that $|\{i:\beta_i\leq t\}|\geq k$ and $|\{i:\alpha_i\geq t\}|\geq n+2-k$.

Let $c=n+1-k$. 
Let $S\in\Sigma^N$ be such that $\sup S_i < t$ for $i\leq c$ and such that $\inf S_i > t$ for $i>c$.
Then $\inf \SO = \underline{g}(\inf S_1,\hdots,\inf S_n)=m(\inf S_1,\hdots,\inf S_n,\alpha_{(1)},\hdots{},\alpha_{(n+1)})$ and $\sup \SO=\overline{g}(\sup S_1,\hdots,\sup S_n) = m(\sup S_1,\hdots,\sup S_n,\beta_{(1)},\hdots{},\beta_{(n+1)})$.
As $m(\inf S_1,\hdots,\inf S_n,\alpha_{(1)},\hdots{},\alpha_{(n+1)})\geq t \geq m(\sup S_1,\hdots,\sup S_n,\beta_{(1)},\hdots{},\beta_{(n+1)})$, it follows that $\inf \SO \geq \sup \SO$, a contradiction that proves the claim.
\end{proof}

\begin{proof}[Proof of Theorem~\ref{thm:tisp}]
That endpoint rules satisfy the axioms follows from Theorem~\ref{mainthm}, Claim~\ref{claim1}, and the fact that ordinal neutrality implies translation equivariance.
Let $f$ satisfy the three axioms.

Because $f$ satisfies anonymity and strategyproofness, it follows from Claim~\ref{claim1} that there exists $P\in\extSigma^{n+1}$ such that $\SO = \med(S\oplus P)$ for all $S\in\Sigma^N$.

I first show that $\inf P_j\in\{-\infty,\infty\}$ for all $j\leq n+1$.
Suppose by means of contradiction that there exists $j\leq n+1$ such that $\inf P_j\in\mathbb{R}$, and let $c\equiv\inf P_j$.
Let $k\equiv|\{j'\leq n+1:\inf P_{j'} < c\}|$. 
Because $\inf P_j=c$, it follows that $k\leq n$.

Let $S\in\Sigma^N$ be such that $S_i=(c-1,c)$ for $i\leq n-k$ and $S_i=(c+1,c+2)$ for $i>n-k$.
Let $t\in(-1,1)$.
It follows that $\inf (S_i+t)<c$ for $i\leq n-k$ and $\inf (S_i+t) > c$ for $i>n-k$.
Then $|\{i\leq n+1:\inf P_i < c\}|+|\{i\in N:\inf (S_i+t)<c\}|=n$ and $|\{i\leq n+1:\inf P_i \leq c\}|+|\{i\in N: \inf (S_i+t)\leq c\}|\geq n+1$.
It follows that, for $x\in\mathbb{R}$, $|\{i\leq 2n+1:\inf((S+t)\oplus P)_i<x\}|\geq n+1$ if and only if $x>c$.

Because $f(S+t)=\med((S+t)\oplus P)$, this implies $f(S+t)\subseteq(c,\infty)$.
Because $f(S+t)\in\Sigma$, there exists $x\in f(S+t)$.
For all $x' < x$, $\sup((S+t)\oplus P)_i>x$ implies that $\sup((S+t)\oplus P)_i>x'$.
Therefore, $|\{i\leq 2n+1:\sup((S+t)\oplus P)_i>x'\}|\geq|\{i\leq 2n+1:\sup((S+t)\oplus P)_i>x\}|\geq n+1$.
It follows that $x'\in f(S+t)$ for all $x'\in(c,x)$, and thus $(c,x)\subseteq f(S+t)$.
Therefore, $\inf f(S+t)=c$ for all $t\in(-1,1)$.
By translation equivariance, $\inf f(S+\tfrac{1}{2})=\inf\SO+\tfrac{1}{2}=c+\tfrac{1}{2}$, a contradiction.
A similar argument shows that $\sup P_j\in\{-\infty,\infty\}$ for all $j\leq n+1$.
Because $\inf P_j\extlt\sup P_j$ for all $j\leq n+1$, it follows that $P_j\in\{(-\infty,-\infty),(-\infty,\infty),(\infty,\infty)\}$ for all $j\leq n+1$.

Let $p\equiv|\{j\leq n+1:P_j=(\infty,\infty)\}|$ and $q\equiv|\{j\leq n+1:P_j=(-\infty,-\infty)\}|$, and note that $p+q\leq n+1$.
Let $T\in\Sigma^N$ and $x\in\mathbb{R}$.
Because $\inf P_j<x$ if and only if $P_j\neq(\infty,\infty)$, and $\sup P_j>x$ if and only if $P_j\neq(-\infty,-\infty)$, it follows that $|\{i\leq 2n+1:\inf(T\oplus P)_i<x\}|=|\Gp{T}{x}|+n+1-p$ and $|\{i\leq 2n+1:\sup(T\oplus P)_i>x\}|=|\Gq{T}{x}|+n+1-q$.
Therefore, $|\{i\leq 2n+1:\inf(T\oplus P)_i<x\}|\geq n+1$ if and only if $|\Gp{T}{x}|\geq p$, and $|\{i\leq 2n+1:\sup(T\oplus P)_i>x\}|\geq n+1$ if and only if $|\Gq{T}{x}|\geq q$.
It follows that $f(T)=\med(T\oplus P)=\{x\in\mathbb{R}:|\Gp{T}{x}|\geq p\textnormal{ and }|\Gq{T}{x}|\geq q\}$.
Because $f(T)$ is non-empty and bounded, $p,q\geq 1$.
Therefore, $f=f^{p,q}$ is an endpoint rule.
\end{proof}

\subsubsection*{Declaration of generative AI and AI-assisted technologies in the manuscript preparation process}

During the preparation of this work, the author used Claude, ChatGPT, and refine.ink in order to simplify proofs, check for errors, identify relevant references, and assist with editing. After using these tools, the author reviewed and edited the content as needed and takes full responsibility for the content of the published article.

\bibliography{ReasonableMan} 
\bibliographystyle{ecta}

\end{document}